\documentclass[a4paper]{article}
\pdfoutput=1
\usepackage{graphicx}
\usepackage{amsmath}
\usepackage{amsfonts}
\usepackage{amssymb}

\newtheorem{theorem}{Theorem}
\newtheorem{claim}{Claim}
\newtheorem{corollary}[claim]{Corollary}
\newtheorem{definition}{Definition}

\newtheorem{assumption}{Assumption}
\newtheorem{lemma}[claim]{Lemma}
\newtheorem{observation}{Observation}

\newenvironment{proof}[1][Proof]{\noindent\emph{#1.} }{\ $\Box$ \medskip}

\begin{document}

\title{Strong Stability of Nash Equilibria \\ in Load Balancing Games%
\footnote{Preprint. To appear in \emph{Science China Mathematics}, Science China Press and
Springer-Verlag Berlin Heidelberg}}
\author{Bo CHEN%
\footnote{Corresponding author: Centre for Discrete Mathematics and Its Applications,
Warwick Business School, University of Warwick, Coventry CV4 7AL, UK;
\texttt{b.chen@warwick.ac.uk}}
\\
University of Warwick, United Kingdom
\and Song-Song LI%
\footnote{School of Management, Qufu Normal University, Rizhao 276826, China;
\{\texttt{163lisongsong, yuzhongrz}\}\texttt{@163.com}}
\ \ \ Yu-Zhong ZHANG${}^{\ddagger}$ \\
Qufu Normal University, China
}

\date{6 November 2013}

\maketitle

\begin{abstract}
We study strong stability of Nash equilibria in load balancing
games of $m$ ($m \ge 2$) identical servers, in which every job
chooses one of the $m$ servers and each job wishes to minimize its
cost, given by the workload of the server it chooses.

A Nash equilibrium (NE) is a strategy profile that is resilient to
unilateral deviations. Finding an NE in such a game is simple.
However, an NE assignment is not stable against coordinated
deviations of several jobs, while a strong Nash equilibrium (SNE) is.
We study how well an NE approximates an SNE.

Given any job assignment in a load balancing game, the improvement
ratio (IR) of a deviation of a job is defined as the ratio between
the pre- and post-deviation costs. An NE is said to be a $\rho$-approximate
SNE ($\rho\geq1$) if there is no coalition of jobs such
that each job of the coalition will have an IR more than $\rho$ from
coordinated deviations of the coalition.

While it is already known that NEs are the same as SNEs in the $2$-server load balancing game,
we prove that, in the $m$-server load balancing game for any given $m\geq3$, any
NE is a $({5}/{4})$-approximate SNE, which together with the lower bound already established
in the literature yields a tight approximation bound. This closes the final gap
in the literature on the study of approximation of general NEs to SNEs in
load balancing games. To establish our upper bound, we make a novel use of a graph-theoretic tool.

\end{abstract}

\section{Introduction}

In game theory, a fundamental notion is \emph{Nash equilibrium} (NE), which is
a state that is \emph{stable} against deviations of any individual participants (known as agents) of the game
in the sense that any such deviation will not bring about additional benefit to the deviator. Much stronger stability
is exhibited by a \emph{strong Nash equilibrium} (SNE), a notion introduced by Aumann \cite{Aum59},
at which no coalition of agents exists such that each member of the coalition can benefit
from coordinated deviations by the members of the coalition.

Evidentally selfish individual agents stand to benefit from cooperation and hence
SNEs are much more preferred to NEs for stability. However, SNEs do not necessarily exist \cite{AnFeMa07} and,
even if they do, they are much more difficult to identify and to compute \cite{FeTa09,Chen09}.
It is therefore very much desirable to have the advantages of both computational efficiency and
strong stability, which motivates our study in this paper. We establish that, for
general NE job assignments in load balancing games,
which exist and are easy to compute, their loss of strong stability
possessed by SNEs is at most 25\%.

In a load balancing game, there are $n$ selfish agents, each representing one of a set
$J=\{J_{1},\cdots,J_{n}\}$ of $n$ jobs. In the absence of a coordinating authority, each agent must choose
one of $m$ identical servers, $M=\{1, \ldots, m\}$, to assign his job to
in order to complete the job as soon as possible.
All jobs assigned to the same server will finish at the same time, which
is determined by the workload of the server, defined to be the total
processing time of the jobs assigned to the server.
Let job $J_{j}$ have a processing time
$p_{j}$ ($1\leq j\leq n$) and let $S_{i}$ denote the set of jobs
assigned to server $i$ ($1\leq i\leq m$). For convenience, we will use
``agent" and ``job'' interchangeably, and consider job processing
times also as their ``lengths''. The completion time $c_{j}$ of job
$J_{j}\in S_{i}$ is the \emph{workload} of its server: $L_{i}=\sum_{J_{j}
\in S_{i}}p_{j}$.

The notions of NE and SNE can be stated more specifically for the load balancing game.
A job assignment $S = (S_1, \ldots, S_m)$ is said to be an NE if no individual
job $J_j \in S_i$ can reduce its completion time $c_j$ by unilaterally migrating
from server $i$ to another server. A job assignment $S = (S_1, \ldots, S_m)$ is said to be an SNE if
no subset $\Gamma\subseteq J$ of jobs can each reduce their job completion times by forming a coalition
and making coordinated migrations from their own current servers.

NEs in the load balancing game have been widely studied (see, e.g.,
\cite{FiHo79,KoPa99,CzVo02,KoMaSp03,ChGu12})
with the main focus of quantifying their loss of global optimality in terms of the price of
anarchy, a term coined by Koutsoupias and Papadimitriou \cite{KoPa99}, as largely summarized
in \cite{Vocking07}. In this paper, we study NEs in load balancing games
from a different perspective by
quantifying their loss of strong stability.

We focus on pure NEs, those corresponding to deterministic job assignments in load
balancing games. While high-quality NEs are easily computed, identification of an SNE is strongly NP-hard \cite{Chen09}.
Given any job assignment in a load balancing game, the \emph{improvement
ratio} (IR) of a deviation of a job is defined as the ratio
between the pre- and post-deviation costs. An NE is said to be a
$\rho$-approximate SNE ($\rho\geq1$) (which is called $\rho$-SE in \cite{Albers09})
if there is no coalition of
jobs such that each job of the coalition will have an IR more than
$\rho$ from coordinated deviations of the coalition. Clearly, the
stability of NE improves with a decreasing value of $\rho$ and a
$1$-approximate SNE is in fact an SNE itself.

For the load balancing game of two servers, one can easily verify
that every NE is also an SNE \cite{AnFeMa07}. If there are three or four servers in
the game, then it is proved in \cite{FeTa09} and \cite{Chen09}, respectively,
that any NE assignment is a $(5/4)$-approximate SNE, and the bound is
tight. Furthermore, it is a $(2-{2}/(m+1))$-approximate SNE if the game has
$m$ servers for $m \ge 5$ \cite{FeTa09}.

We establish in this paper that, in the $m$-server load balancing game
($m\geq3$), any NE is a $(5/4)$-approximate SNE, which is tight and hence
closes the final gap in the literature on the study of NE approximation of SNE in load balancing games.
To establish our approximation bound, we make a novel use of a powerful graph-theoretic tool.

\section{Definitions and Preliminaries}

\subsection{A Lower Bound}
\label{sec:lower_bound}

We start with an example to help the reader get some intuition of the problem under consideration. The example also provides a lower bound of $5/4$ for any NE assignment to approximate SNE. The left panel of Fig.~\ref{fig:lower_bound} below shows an NE assignment of six jobs to three identical machines with job completions $5$, $5$ and $10$, respectively, for the three pairs of jobs. If the four jobs of lengths $2$ and $5$ form a coalition and make a coordinated deviation as shown in the figure, then in the resulting assignment, each of the four jobs in the coalition achieves an improvement ratio of $5/4$.

\begin{figure}[h!]
\centering
\includegraphics[scale=0.8,viewport=40mm 235mm 165mm 270mm,clip=true]{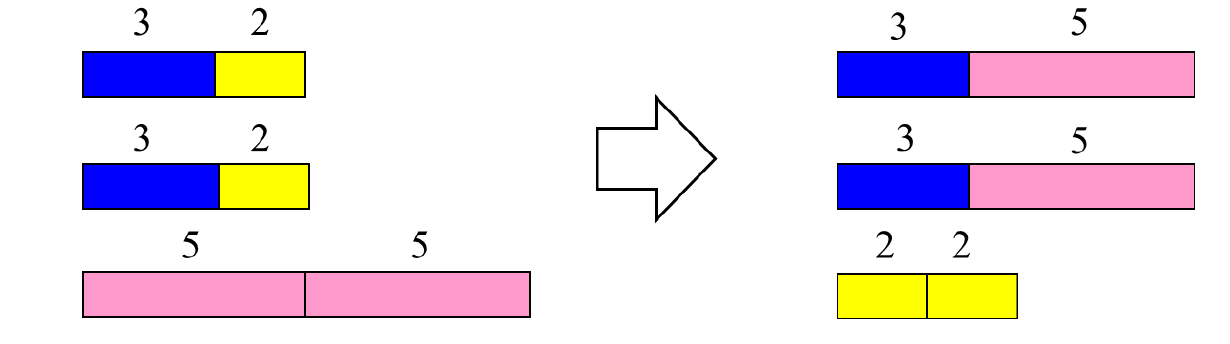} \\  %
  \caption{An Instance for Lower Bound}\label{fig:lower_bound}
\end{figure}

\subsection{Graph-theoretic Tool~\cite{Chen09}}

As a tool of our analysis, we start with the minimal deviation graph introduced by Chen \cite{Chen09}.
For convenience we collect into this subsection some basic results on minimal deviation graphs
from \cite{Chen09}. Given an NE job assignment $S=(S_1, \ldots, S_m)$, as an NE-based
coalitional deviation or simply \emph{coalitional deviation} $\Delta$, we refer to a
collective action of a subset $\Gamma \subseteq J$ of jobs in which each job of $\Gamma$ migrates from
its server in the assignment $S$ so that its completion time is decreased after the migration. Accordingly, $\Gamma = \Gamma(\Delta)$ is called the corresponding \emph{coalition}. We
introduce deviation graphs to characterize coalitional deviations.
In a coalitional deviation, a server $i$ is said to be \emph{participating}
or \emph{involved} if its job set changes after the deviation. Given a
coalitional deviation $\Delta$ with the corresponding coalition
$\Gamma=\Gamma(\Delta)$, we define the corresponding (directed) deviation graph
$G(\Delta)=(V,A)$ as follows:
\begin{eqnarray*}
  V = V(G)&:=&\{i:\ \textrm{server $i$ is a participating server}\}; \\
  A = A(G)&:=&\{(u,v):\ \textrm{a job $J_{j}\in\Gamma$
migrates from $S_{u}$ to $S_{v}$}\}.
\end{eqnarray*}
In what follows, without loss of generality we consider coalitional deviations with
$V(G)=M$. Given a coalitional deviation $\Delta$, we denote by $L'_{i}=L_{i}(\Delta)$
the workload of server $i$ after deviation $\Delta$, and by IR($\Delta$) the minimum of the
improvement ratios of all jobs taking part in $\Delta$.
Then we have the following definition and lemmas from~\cite{Chen09}:

\begin{lemma}\label{lem:out-degree}
The out-degree $\delta^{+}(i)$ of any node $i$
of a deviation graph is at least 1, and hence $|S_{i}|\geq2$.
\end{lemma}

\begin{lemma}\label{lem:no-cycles}
If all $m$ servers are involved in a coalitional deviation, then the deviation graph does not contain a set of node-disjoint directed cycles such that each node of the graph is in one of the directed cycles.
\end{lemma}

\begin{definition}\label{def:min-deviation}
Let $\Delta$ be a coalitional deviation and $\Gamma = \Gamma(\Delta)$ be the corresponding coalition. Deviation graph $G=G(\Delta)$ is said to be \emph{minimal} if $\textrm{IR}(\Delta')<\textrm{IR}(\Delta)$ for
any coalitional deviation $\Delta'$ such that the corresponding coalition $\Gamma'=\Gamma'(\Delta')$ is a proper subset
of $\Gamma$.
\end{definition}

\begin{lemma}\label{lem:in-degree}
The in-degree $\delta^{-}(i)$ of any node $i$ of a
minimal deviation graph is at least 1.
\end{lemma}

\begin{lemma}\label{lem:strong-connectivity}
A minimal deviation graph is strongly connected.
\end{lemma}

\subsection{Some Observations}

In our study of bounding NE approximation of SNE, we can apparently focus on those
coalitional deviations that correspond to minimal deviation graphs.
We start with several observations on any NE-based coalitional deviation
$\Delta$ involving $m$ servers for $m\geq3$. Let $G(\Delta)$ denote the corresponding
minimal deviation graph.

If two jobs assigned to server $i \in M$ in the NE
assignment migrate to server $j \in M$ ($j\neq i$) together, or
both stay on the server, then we can treat them as one single job without
loss of generality in our study of the minimal deviation graph. With this
understanding, if we let
$a_{i}$ ($i \in M$) denote the number of jobs assigned to server
$i$ in the NE assignment, then the following is immediate.
\begin{observation}\label{obs:out-degrees}
For any $i \in M$, we have $2\leq a_{i}\leq m$. $\delta^{+}(i)=a_{i}$ or
$\delta^{+}(i)=a_{i}-1$.
\end{observation}
As a result of the above observation, the node set $M$ can be partitioned into
two, $M'$ and $M''$, as follows:
\begin{eqnarray*}
  M' &:=& \{i\in M:\ a_i=\delta^+(i)\}, \\
  M'' &:=& M\backslash M'=\{i\in M:\ a_i=\delta^+(i)+1\}.
\end{eqnarray*}
By applying a data scaling if necessary, we assume without loss of generality that
\begin{equation}\label{eq:scalling}
    \min_{i \in M}L_{i}=1.
\end{equation}

\begin{observation}\label{obs:bound-of-Li}
For any $i\in M$, we have $L_{i}\leq {a_{i}}/({a_{i}-1})$.
\end{observation}

\begin{proof}
Suppose to the contrary that $L_{i}>{a_{i}}/({a_{i}-1})$, which
implies that $a_{i}>{L_{i}}/({L_{i}-1})$.

Let $x_{i}$ denote the length of the shortest job assigned to server
$i$ in the NE assignment. We have $L_{i}\geq a_{i}x_{i}$, which leads to $L_{i}>{L_{i}x_{i}}/({L_{i}-1})$, that is,
$L_{i}>x_{i}+1$, which implies that the shortest job assigned to
server $i$ in the NE assignment can have the benefit of reducing its
job completion time by unilaterally migrating to the server of which
the workload is 1, contradicting the NE property.
\end{proof}

The following observation states that, if all jobs on a server participate in the migration,
then none of the servers they migrate to will have \emph{all}
its jobs migrate out.

\begin{observation}\label{obs:out-degree-relation}
If $(i,j)\in A$ and $i\in M'$, then $j\in M''$.
\end{observation}

\begin{proof}
Suppose to the contrary that $a_{j}\neq\delta^{+}(j)+1$. According
to Observation~\ref{obs:out-degrees}, we have $a_{j}=\delta^{+}(j)$, which implies
that all the jobs assigned to server $i$ and server $j$ in the NE
assignment belong to coalition $\Gamma$.

Since $(i,j)\in A$, there is a job $J_{k}\in\Gamma$ that migrates from server $i$ to server $j$.
Consider the new coalition $\Gamma'$ formed by all members of $\Gamma$ except
$J_{k}$. Then we have $\emptyset \neq \Gamma'\subset\Gamma$. Let $\Delta'$ be such a coalitional deviation
of $\Gamma'$ that is the same as $\Delta$ except without the involvement of $J_{k}$ and the job(s) that
migrate(s) to $i$ (resp.~$j$) in $\Delta$ will migrate to $j$ (resp.~$i$) in $\Delta'$. Then we
have $\textrm{IR}(\Delta')=\textrm{IR}(\Delta)$, contradicting the minimality of the deviation
graph $G$ according to Definition~\ref{def:min-deviation}.
\end{proof}

\smallskip

The following observation is a direct consequence of Observation~\ref{obs:out-degree-relation}:

\begin{observation}\label{obs:equiv-min-deviation}
Assume $i,j \in M'$. Hence $(i,j),(j,i) \not\in A$ according
to Observation~\ref{obs:out-degree-relation}. Let $\Delta'$ be the same as $\Delta$
except that any job that migrates to $i$ (resp.~$j$) in $\Delta$ will migrate to $j$
(resp.~$i$) in $\Delta'$. Then $\textrm{IR}(\Delta')=\textrm{IR}(\Delta)$, and $G(\Delta')$ is also
minimal.
\end{observation}

\section{A Key Inequality}

To help our analysis, we will introduce in this section a special arc set
$\widetilde{A}\subseteq A$ in the minimal deviation graph $G(\Delta)$.

\subsection{Auxiliary Arc Set $\widetilde{A}$}
\label{sec:auxiliary_arc_set}

For any node $i \in M$, denote $Q^{+}(i) :=\{j\in M:\ (i, j) \in A \}$ and
$Q^{-}(i):=\{j \in M:\ (j, i) \in A\}$. For notational convenience, for any
node set $S\subseteq M$, we denote $Q^+(S):=\bigcup_{i\in S}Q^+(i)$ and
$Q^-(S):=\bigcup_{i\in S}Q^-(i)$. With $A$ replaced by $\widetilde{A}$ above,
we similarly define $\widetilde{Q}^{+}(i)$, $\widetilde{Q}^{-}(i)$,
$\widetilde{Q}^{+}(S)$ and $\widetilde{Q}^{-}(S)$.

Let us define $\widetilde{A}$ as follows. According to Lemma~\ref{lem:in-degree},
$|Q^{-}(i)|\geq 1$ for
any $i\in M$. For each $i\in M$ we pick up an arc from the non-empty set $Q^{-}(i)$ to form an $m$-element
subset $\widetilde{A} \subseteq A$. Then $\widetilde{A}$ possesses the following property:
\begin{equation}\label{eqn:in-degree}
    |\widetilde{Q}^-(i)| = 1 \textrm{ for any } i \in M.
\end{equation}
Denote $b_i := |\widetilde{Q}^+(i)|$ for any $i\in M$. Then it is clear that
\begin{equation}\label{eqn:out-degree}
    \sum_{i=1}^m b_i = |\widetilde{A}|= m.
\end{equation}
If node set $S$ is a singleton, then we will also use $S$ to denote the singleton if no confusion
can arise. Hence, due to (\ref{eqn:in-degree}) we will also use $\widetilde{Q}^-(i)$ to denote the single
\emph{element} of the corresponding set. Any arc set $\widetilde{A} \subseteq A$
that possesses property (\ref{eqn:in-degree}) is said to be \emph{tilde-valid}.

\subsection{Main Result}

Our main result is stated in the following theorem.
\begin{theorem}\label{thm:main}
For any minimal deviation graph $G(\Delta_m)$ involving $m$ servers, its improvement ratio
$\textrm{IR}(\Delta_{m})\leq {5}/{4}$.
\end{theorem}

Let us perform some initial investigation to see what we need to do to prove the theorem.
Recall that, for any $i\in M$, $a_{i}$ is the number of jobs assigned to
server $i$ in the NE assignment and $b_i = |\widetilde{Q}^+(i)|$ for a fixed
arc set $\widetilde{A}$ defined in Section~\ref{sec:auxiliary_arc_set} for the minimal deviation graph $G(\Delta_m)$. For a pair of integers $a$ and $b$ with $2\leq a\leq m$ and $0\leq b\leq
a$, let $M_{a}^{b}:=\{i\in M:\ a_{i}=a,b_{i}=b\}$. Then it is clear that
\begin{equation}\label{eqn:def-of-Mab}
    \bigcup_{2\leq a\leq m}\bigcup_{0\leq b\leq a} M_{a}^{b} = M.
\end{equation}
Denote $m_{a}^{b}=|M_{a}^{b}|$ for all possible pairs $a$
and $b$: $2\le a\le m$ and $0\le b \le a$. Let $r=\textrm{IR}(\Delta_{m})$. Then according to (\ref{eqn:in-degree})
and (\ref{eqn:def-of-Mab}), we have
\begin{equation}\label{eqn:equations-for-Dab}
    \sum\limits_{a=2}^{m}\sum\limits_{b=0}^{a}m_{a}^{b}=m, \ \textrm{ and } \
    \sum\limits_{a=2}^{m}\sum\limits_{b=0}^{a}b\,m_{a}^{b}=m.
\end{equation}
According to the definition of IR, we have $rL_{j}'\leq L_{i}$ for $(i,j) \in A$.
Summing up these inequalities over all $m$ arcs in $\widetilde{A}$ leads to
\[
\sum_{j=1}^{m}rL_{j}'\leq\sum_{i=1}^{m}b_{i}L_{i},
\]
which implies that
\begin{equation}\label{eqn:bound-of-r}
    r\leq\frac{\sum\limits_{i=1}^{m}b_{i}L_{i}}{\sum\limits_{i=1}^{m}L_{i}}.
\end{equation}
According to Observation~\ref{obs:bound-of-Li}, we have $L_{i}\leq{a_{i}}/({a_{i}-1})\leq2$,
which implies that the right-hand side of (\ref{eqn:bound-of-r}), which we denote by $R$, is
at most 2, since $R$ is a convex combination of $L_1, \ldots, L_m$ and $0$ with the corresponding combination coefficients $\lambda_i=b_i/\sum_{k=1}^m L_k$ ($i=1, \ldots, m$) and $\lambda_{m+1}=1- m/\sum_{k=1}^m L_k\ge 0$ due to (\ref{eq:scalling}) and (\ref{eqn:out-degree}). On the other hand,
since $r\ge 1$ according to the definition, we conclude that $1\le R\le 2$, which
implies that $R$ is a decreasing function of $L_{i}$ for which $b_{i}=0$ or $b_{i}=1$,
and an increasing function of $L_{i}$ for which $b_{i}\geq2$. Therefore, we increase $R$ by
increasing $L_{i}$ to ${a_{i}}/({a_{i}-1})$ for $i$ such that $b_{i}\geq2$, and by decreasing
$L_{i}$ to $1$ for $i$ such that $b_{i}=0$ or $b_{i}=1$. Noticing that $2\le a_i\le m$ according to Observation~\ref{obs:out-degrees}, we obtain
$$
r\leq\frac{\sum\limits_{a=2}^{m}\sum\limits_{b=2}^{a}\frac{ab}{a-1}m_{a}^{b}
+\sum\limits_{a=2}^{m}m_{a}^{1}}{\sum\limits_{a=2}^{m}\sum\limits_{b=2}^{a}\frac{a}{a-1}m_{a}^{b}
+\sum\limits_{a=2}^{m}m_{a}^{1}+\sum\limits_{a=2}^{m}m_{a}^{0}},
$$
which together with (\ref{eqn:equations-for-Dab}) implies that
$$
r\leq\frac{\sum\limits_{a=2}^{m}\sum\limits_{b=2}^{a}
\frac{b}{a-1}m_{a}^{b}+m}{\sum\limits_{a=2}^{m}
\sum\limits_{b=2}^{a}\frac{1}{a-1}m_{a}^{b}+m}.
$$
In order to prove Theorem~\ref{thm:main}, we need to show $r\leq {5}/{4}$. Then it suffices to show
$$
\frac{\sum\limits_{a=2}^{m}\sum\limits_{b=2}^{a}
\frac{b}{a-1}m_{a}^{b}+m}{\sum\limits_{a=2}^{m}
\sum\limits_{b=2}^{a}\frac{1}{a-1}m_{a}^{b}+m}\leq\frac{5}{4},
$$
which is equivalent to
$$
\sum_{a=2}^{m}\sum_{b=2}^{a}\frac{4b-5}{a-1}m_{a}^{b}\leq m.
$$
By replacing the right-hand side $m$ of the above inequality with the left-hand side of the second equality in (\ref{eqn:equations-for-Dab}), we have
$$
0\leq\sum_{a=2}^{m}m_{a}^{1}+\sum_{a=2}^{m}\sum_{b=2}^{a}
\left(b-\frac{4b-5}{a-1}\right)m_{a}^{b},
$$
that is
\begin{equation}\label{eqn:ultimate}
m_{2}^{2}+\frac{1}{2}m_{3}^{3}\leq\sum_{a=2}^{m}m_{a}^{1}+\frac{1}{2}m_{3}^{2}
+\sum_{a=4}^{m}\sum_{b=2}^{a}\left(b-\frac{4b-5}{a-1}\right)m_{a}^{b}.
\end{equation}

In what follows, we are to prove (\ref{eqn:ultimate}) and thereby Theorem~\ref{thm:main} through a series of lower bounds established in Section~\ref{sec:bounding} on different terms of the right-hand side of inequality~(\ref{eqn:ultimate}).

\section{Preparations}

We introduce an auxiliary node set $W$ in addition to the auxiliary arc set $\widetilde{A}$ introduced earlier.

\subsection{Auxiliary Node Set $W$}
\label{sec:auxiliary_node_set}

Let
\[
W_{0}:= \{ i\in M: b_i = 0\}.
\]
Then we immediately have
\begin{lemma}\label{lem:W0}
$W_{0}\neq\emptyset$.
\end{lemma}
\begin{proof}
Suppose to the contrary that $W_{0}=\emptyset$. Then any $b_i \ge 1$ in (\ref{eqn:out-degree}), which implies that
$b_{i}=1$ for any $i \in M$, so that $\widetilde{A}$ forms some node-disjoint directed cycles
that span all nodes, contradicting Lemma~\ref{lem:no-cycles}.
\end{proof}

Note that, from the formation of arc set $\widetilde{A}$, it is clear that $\widetilde{A}$ as
a tilde-valid arc set may not be unique. However, among all possible choices of a tilde-valid
arc set $\widetilde{A} \subseteq A$, we choose one that has some additional properties in terms of
minimum cardinalities of some combinatorial structures, which we shall define in due course.
These additional properties will be presented in a sequence of three assumptions, which are made without
loss of generality due to the finiteness of the total number of tilde-valid arc sets.
With the same reason, we assume that our coalitional deviation $\Delta$ is chosen in such a way that
it has a certain property (see Assumption~\ref{ass:Delta}).
\begin{assumption}\label{ass:mimimum-W0}
    Arc set $\widetilde{A}$ is tilde-valid and it minimizes $|W_0(\widetilde{A})|$.
\end{assumption}

Let $\widetilde{W}_0 = Q^+(W_0)$. Then $\widetilde{W}_0 \neq \emptyset$ according to Lemmas~\ref{lem:W0} and
\ref{lem:out-degree}. A node $i\in M$ is said to be \emph{associated} with $W_0$ if it is \emph{linked} to
an element of $\widetilde{W}_0$ through a sequence of arcs (but not a directed path) in $\widetilde{A}$ and $A$ in alternation (see Fig.~\ref{fig:W_1} for an illustration). More
formally, $i\in M$ is associated with $W_0$ if and only if, for some integer $k\ge 0$, there are nodes $\{i_0,
\ldots, i_k, j_0, \ldots, j_k\} \subseteq M$ with $i = i_k$ and $j_0 \in \widetilde{W}_0$, such that
\begin{equation}\label{eqn:def-of-W1}
(i_0,j_0), \ldots, (i_k,j_k) \in \widetilde{A} \ \ \textrm{ and  } \ \ (i_0,j_1), \ldots, (i_{k-1},j_k) \in A.
\end{equation}

\begin{figure}[h!]
\centering
NB: the solid arcs belong to $\widetilde{A}$ and dotted arcs to $A$ \\
\includegraphics[scale=0.6,viewport=30mm 220mm 165mm 280mm,clip=true]{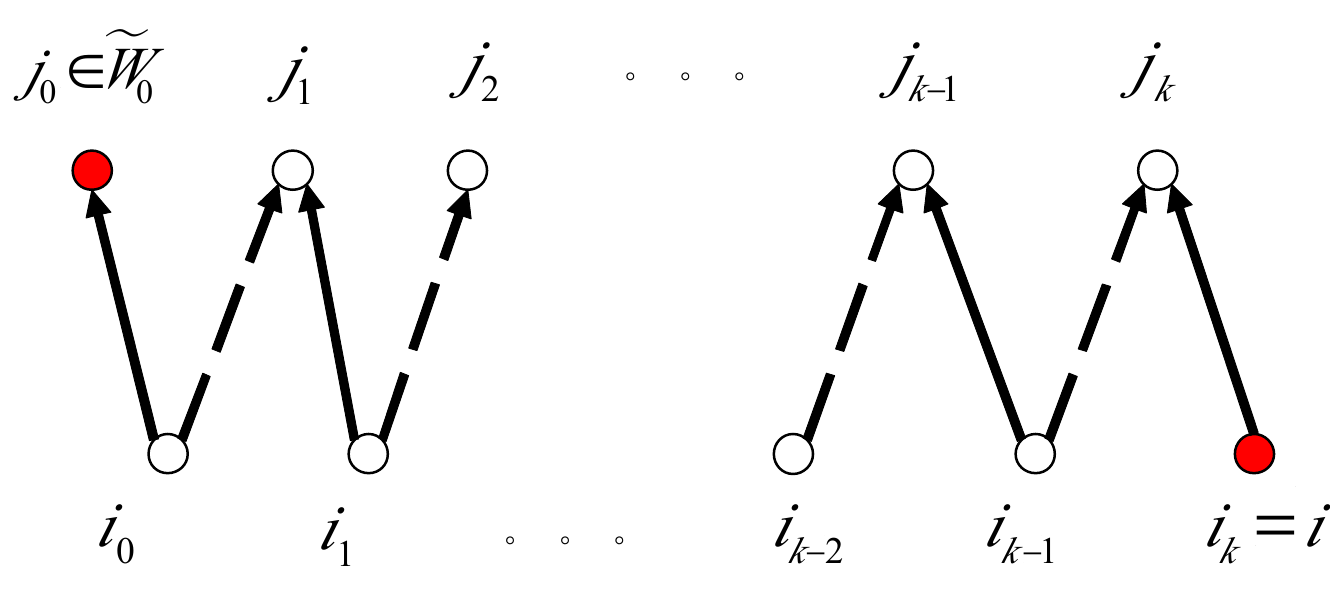}
  \caption{Definition of Node Set $W_1$}\label{fig:W_1}
\end{figure}

\noindent
Note that in the above definition, if $i=i_k$ is associated with $W_0$, then $i_0$,
\ldots, $i_{k-1}$ used in (\ref{eqn:def-of-W1}) are each associated with $W_0$. Define
\begin{eqnarray*}
  W_1 &:=& \{ i\in M:\ \textrm{node  $i$  is associated with }  W_0 \}, \\
  \widetilde{W}_1 &:=& \widetilde{Q}^+(W_1).
\end{eqnarray*}
Immediately we have
$\widetilde{Q}^-(\widetilde{W}_0)\subseteq W_1$, which implies that
\begin{equation}\label{eqn:tilde-W0-versus-tilde-W1}
\widetilde{W}_0 \subseteq \widetilde{W}_1.
\end{equation}
On the other hand, since $\widetilde{Q}^-(\widetilde{W}_0) \neq \emptyset$ according to (\ref{eqn:in-degree}),
we have $W_1\neq \emptyset$.

\begin{lemma}\label{lem:W1}
    For any $i\in W_1$, $b_i =1$. Furthermore, $Q^+(W_0 \cup W_1) = \widetilde{W}_1$.
\end{lemma}
\begin{proof}
It is clear from the definition that $W_1\cap W_0 = \emptyset$. Hence $b_i \ge 1$ for any $i\in W_1$. Assume
for contradiction that $b_i \ge 2$ for some $i\in W_1$. Since $i$ is associated with $W_0$, in addition to nodes
$\{i_0,\ldots, i_k, j_0, \ldots, j_k\} \subseteq M$ satisfying
(\ref{eqn:def-of-W1}), we have a node $h \in W_0$ (hence $h \not\in \{i_0, \ldots, i_k\}$)
such that $(h,j_0)\in A$ according to the definition
of $\widetilde{W}_0$. Now we remove $k+1$ arcs $(i_0,j_0), \ldots, (i_k,j_k)$ from $\widetilde{A}$ and add $k+1$
new arcs $(h,j_0), (i_0,j_1), \ldots, (i_{k-1},j_k)$ to $\widetilde{A}$. It is easy to see that the new set
$\widetilde{A}$ still has property (\ref{eqn:in-degree}). Additionally, under the new
$\widetilde{A}$, all $\{b_k\}$ remain the same except two of them: $b_h$ and $b_i$, with the former increased
by 1 and the latter decreased by 1. Since $b_i \ge 2$ under the original $\widetilde{A}$, then $i\not\in W_0$ under
the new $\widetilde{A}$. Consequently, the new $W_0$ determined by the new $\widetilde{A}$ contains a smaller number
of elements, contradicting Assumption~\ref{ass:mimimum-W0} about the original $\widetilde{A}$.

To prove the second part of the lemma, let us first prove $Q^+(W_1)\subseteq \widetilde{W}_1$. Let $i \in W_1$
and $(i,j) \in A$. We show that $j\in \widetilde{W}_1$. In fact, since
$|\widetilde{Q}^-(j)| =1$ according to (\ref{eqn:in-degree}), we have a node
$h \in M$ such that $(h,j) \in \widetilde{A}$. Now since $i$ is associated
with $W_0$, we conclude that $h$ is also associated with $W_0$, which implies
that $j\in \widetilde{W}_1$. Therefore, with (\ref{eqn:tilde-W0-versus-tilde-W1}) we have proved that $Q^+(W_0 \cup W_1)
\subseteq \widetilde{W}_1$. The other direction of the inclusion is apparent.
\end{proof}

It follows from Lemma~\ref{lem:W1} and (\ref{eqn:in-degree}) that the mapping $\widetilde{Q}^+(\cdot)$ from $W_1$
onto $\widetilde{W}_1$ is a one-to-one correspondence and hence
\begin{equation}\label{eqn:cadinality-of-W1}
    |W_1| = |\widetilde{W}_1| > 0.
\end{equation}
Let
\[
W := W_0 \cup W_1\cup\widetilde{W}_1.
\]

\subsection{Notation}

As we can see from inequality~(\ref{eqn:ultimate}), bounding the sizes $m_2^2$ and $m_3^3$ of the respective
sets $M_2^2$ and $M_3^3$ is vital in our establishment of the desired bound. We therefore
take a close look at the two sets by partitioning
$$
X := M_2^2 \cup M_3^3
$$
into a number of subsets, so that different bounding arguments can be applied to
different subsets.

We assemble our notations here in one place for easy reference and the reader is advised to conceptualize each \emph{only} when it is needed in an analysis at a later point.

Let $\widetilde{M}_2^2 := \{\ell\in M_2^2:\ \widetilde{Q}^{+}(\ell)\nsubseteq W\}$.
For convenience, we reserve letter $\ell$ to exclusively index elements of $\widetilde{M}_2^2$ and
let $\widetilde{Q}^+(\ell) = \{\ell_1, \ell_2\}$ with the understanding that it is
always the case that $\ell_1 \notin W$. For any $\ell \in \widetilde{M}_2^2$,
$\ell_1 \notin W$ implies $b_{\ell_1} \ge 1$ since $W_0 \subseteq W $. On the
other hand, since arc $(\ell, \ell_1)\in \widetilde{A}\subseteq A$, we have
$a_{\ell_1} = \delta^+(\ell_1) + 1 \ge b_{\ell_1} +1$ according to
Observation~\ref{obs:out-degree-relation}, which implies that $\ell_1$ must belong to one
of the following three mutually disjoint node sets:
\begin{eqnarray*}
  Z_1 &:=& \{i\in M\backslash W:\ b_i =1\}, \\
  Z_2 &:=& \{i\in M\backslash W:\ b_i > 1, a_i > b_i+1\}, \\
  Z &:=& \{i\in M\backslash W:\ b_{i} > 1, a_{i} = b_{i}+1\}.
\end{eqnarray*}
Therefore, if we define
\[
\left\{
\begin{array}{ccl}
X_2 &:=& \{ i\in M_{3}^{3}:\ \widetilde{Q}^{+}(i)\nsubseteq W \}; \\
X_3 &:=& \{ \ell \in \widetilde{M}_2^2:\ \widetilde{Q}^+(\ell)\cap W = \emptyset\}; \\
X_4 &:=& \{ \ell \in \widetilde{M}_2^2 \backslash X_3:\ \ell_1 \in Z_1\cup Z_2\}; \\
X_5 &:=& \{ \ell \in \widetilde{M}_2^2\backslash X_3:\ \ell_1 \in Z,
          Q^{+}(\ell_1)\nsubseteq W, \delta^{-}(\ell_1)>1 \}; \\
X_6 &:=& \{ \ell \in \widetilde{M}_2^2\backslash X_3:\ \ell_1 \in Z,
          Q^{+}(\ell_1)\nsubseteq W, \delta^{-}(\ell_1)= 1 \};
\end{array}
\right.
\]
then we have
\[
X_{11}:= \widetilde{M}_2^2\setminus{\textstyle\bigcup_{k=3}^6} X_k =
\{\ell\in \widetilde{M}_2^2\backslash X_3:\ \ell_1 \in Z,
            Q^{+}(\ell_1)\subseteq W \},
\]
and
\begin{equation}\label{eqn:bipartition}
    |\widetilde{Q}^+(\ell)\cap W| = |\widetilde{Q}^+(\ell)\backslash W| = 1
\end{equation}
for any $\ell \in \widetilde{M}_2^2\backslash X_3 = X_4\cup X_5\cup X_6\cup X_{11}$. In other words, for any element $\ell\in\widetilde{M}_2^2\backslash X_3$, the two-element set $\widetilde{Q}^+(\ell)$ has exactly one element in $W$. Now let
\[
X_1 := \{ i\in X:\ \widetilde{Q}^{+}(i)\subseteq W \}
                 \cup X_{11}.
\]
Clearly, $X_{i}\cap X_{j}=\emptyset$ ($1\leq i \neq j\leq 6$) and
$X=\bigcup_{k=1}^6 X_{k}$.

\section{Proving Upper Bounds}
\label{sec:bounding}

To bound from below the right-hand side of the key inequality (\ref{eqn:ultimate}),
or equivalently, to bound from above the left-hand side of (\ref{eqn:ultimate}),
we establish through a series of five lemmas and a corollary that the number of nodes in $X_t$ is at most $c_t|Y_t|$, where $c_t \in \{1, \frac12\}$ and $Y_t$ are defined below for $t=1, \ldots, 6$
with $Y_1, \ldots, Y_6 \subseteq M \backslash X$ mutually disjoint node sets. We divide our proofs into two parts with the second part on bounding $|X_6|$.

\subsection{Part 1}

Let us start with some straightforward upper bounds. Since $\widetilde{Q}^{+}(i)\backslash W \neq \emptyset$ for any $i\in X_2$ according to the definition of $X_2$, we immediately have the following lemma thanks to
Observation~\ref{obs:out-degree-relation}.
\begin{lemma}\label{lem:X2}
Let $Y_2:= \bigcup_{i\in X_2} \widetilde{Q}^{+}(i)\backslash W$.
Then $Y_2\subseteq M''\backslash W$ and $|X_{2}|\le |Y_{2}|$. \ $\Box$
\end{lemma}

Note that $|\widetilde{Q}^{+}(\ell)|=2$ for any $\ell\in X_3$ and
$\widetilde{Q}^+(i)\cap \widetilde{Q}^+(j) = \emptyset$ ($i\ne j$) due to (\ref{eqn:in-degree}),
which lead to the following lemma.
\begin{lemma}
    Let $Y_3 := \bigcup_{\ell\in X_3} \widetilde{Q}^{+}(\ell)$. Then $Y_3\subseteq M''
    \backslash W$ and $2|X_{3}| \le |Y_{3}|$. \ $\Box$
\end{lemma}

The following lemma follows directly from the definition of $X_4$:
\begin{lemma}
    Let $Y_4:= \bigcup_{\ell\in X_4} \widetilde{Q}^{+}(\ell)\backslash W $.
    Then $Y_4\subseteq M''\backslash W$ and $|X_{4}| \le |Y_{4}|$. For any $j\in Y_{4}$,
    $b_{j}>1$ and $a_{j} > b_{j}+1$, unless $b_{j}=1$. \ $\Box$
\end{lemma}

At this point, we introduce our second additional assumption about $\widetilde{A}$ without
loss of generality.
\begin{assumption}\label{ass:M-2-2}
    Arc set $\widetilde{A}$ is such that it first satisfies Assumption~\ref{ass:mimimum-W0} and
    then minimizes $|M_2^2(\widetilde{A})|$.
\end{assumption}

For any $ \ell\in X_5$, since $\delta^{-}(\ell_1)>1$
according to the definition of $X_5$, there is $j\in
Q^{-}(\ell_1)\backslash\{\ell\}$. Then $j \notin W_0\cup W_1$ (otherwise we would have
$\ell_1\in W $ according to Lemma~\ref{lem:W1}). In fact, node $j$ has the following
property:
\begin{equation}\label{eqn:2nd-special-node-property}
    j\in M_2^1\cap M'\backslash W.
\end{equation}
To see this, consider replacing $(\ell, \ell_1)$ with $(j,\ell_1)$ in $\widetilde{A}$ to
form a new tilde-valid arc set $\widetilde{A}'$. It is easy to see that $\widetilde{A}'$
satisfies Assumption~\ref{ass:mimimum-W0}. However, with the new arc set $\widetilde{A}'$, $\ell$ is no
longer a node in the new $M_2^2(\widetilde{A}')$, which implies that $j$ has to
become a node in $M_2^2(\widetilde{A}')$ in order not to contradict Assumption~\ref{ass:M-2-2} with the
original choice of $\widetilde{A}$, which in turn implies properties
(\ref{eqn:2nd-special-node-property}).
Furthermore, since $j\in M_{2}^{1}$ and $(j,\ell_1) \in A\backslash \widetilde{A}$,
there is no $k \ne \ell_1$ such that $(j,k) \in A\backslash \widetilde{A}$, which
implies that $j \notin Q^{-1}(\widetilde{Q}^{+}(\ell')\backslash)\backslash\{\ell'\}$.
Consequently, we have the following lemma.
\begin{lemma}\label{lem:X5}
Let $Y_{5}:= \bigcup_{\ell \in X_5}{Q}^-(\widetilde{Q}^+(\ell)\backslash W)\backslash\{\ell\}$.
Then $Y_5\subseteq M_2^1\cap M'\backslash W$ and $|X_{5}| \le |Y_{5}|$. \ $\Box$
\end{lemma}

Now let us establish an upper bound on $|X_1|$ in the following lemma with the minimality of our deviation graph $G(\Delta)$.

\begin{lemma}\label{lem:X1}
  Let $Y_1:=W_1$. Then $|X_{1}|\leq|Y_{1}|$.
\end{lemma}
\begin{proof}
Suppose to the contrary that $|X_{1}|>|W_{1}|$,
that is, $|X_{1}|> |W_{1}|=|\widetilde{W}_1|$ according to (\ref{eqn:cadinality-of-W1}).
Let $X_{1}' \varsubsetneq X_1$ be a proper subset of
$|\widetilde{W}_1| >0$ elements. Define
\begin{eqnarray*}
  K &:=& W \backslash \widetilde{W}_1 \subseteq W_0\cup W_1, \\
  K' &:=& Q^+(X_{1}')\backslash W, \\
  \Gamma' &:=& {\textstyle \{J_{j}\in\Gamma:\ J_{j}\in\bigcup_{i\in\widetilde{M}}S_{i}\},}
\end{eqnarray*}
where $\widetilde{M}:=X_{1}'\cup K\cup K'$ (see Fig.~\ref{fig:lemma_proof} for an illustration with explanations to follow).
\begin{figure}[h!]
\centering
\includegraphics[scale=0.9,viewport=50mm 215mm 165mm 260mm,clip=true]{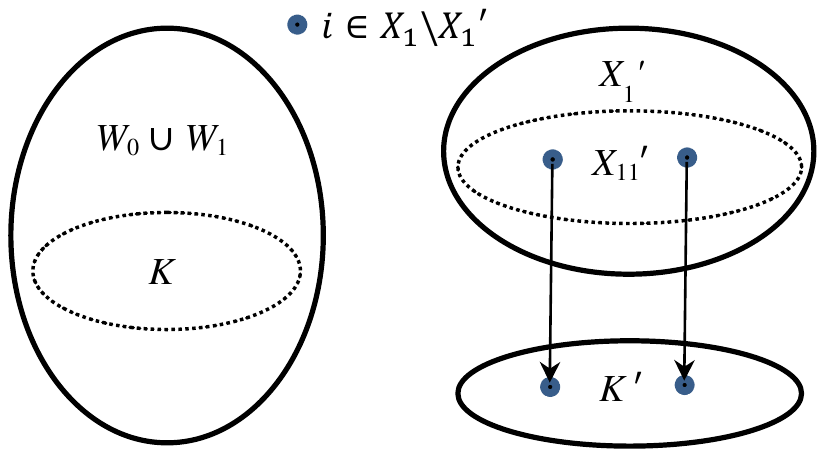}\\ %
  \caption{Proof of Lemma~\ref{lem:X1}}\label{fig:lemma_proof}
\end{figure}

Then $\Gamma'\neq\emptyset$ since $X_1' \neq \emptyset$. We claim $\Gamma'$ is a proper
subset of $\Gamma$. To see this, let $i\in X_{1}\backslash X_1'
\neq \emptyset$. Since $X\cap(W_0\cup W_1)=\emptyset$
(Lemma~\ref{lem:W1}), we have $i\notin K$. Observation~\ref{obs:out-degree-relation}
implies $i\notin K'$. Therefore, we have $i\notin \widetilde{M}$,
i.e., $S_{i}\cap\Gamma'=\emptyset$, but $S_{i}\subseteq\Gamma$. With the same arguments we note that the three constituent subsets of $\widetilde{M}$ are mutually disjoint. In Fig.~\ref{fig:lemma_proof}, the set ${X_{11}}'$ is a subset of $X_{11}$ according to the definition of $X_1$ and the mapping between ${X_{11}}'$ and $K'$ is a one-to-one correspondence due to equation (\ref{eqn:bipartition}).

Since $|X_1'| = |\widetilde{W}_1|$, we can assume there is a one-to-one
correspondence $\phi$ between the nodes (i.e., servers) of the two sets
$X_{1}'$ and $\widetilde{W}_1$. Now let us
define a new coalitional deviation $\Delta'$ with $\Gamma'=\Gamma'(\Delta')$, which is the same
as $\Delta$ restricted on $\Gamma'$ except that, if $J_{j}\in\Gamma'$ migrates
in $\Delta$ to a server of $\widetilde{W}_1$, then let
$J_{j}$ migrate in $\Delta'$ to the corresponding (under $\phi$) server of
$X_{1}'$.

We show that the improvement ratio of any job deviation in $\Delta'$ is at least the
same as that in $\Delta$, which then implies that $\textrm{IR}(\Delta')\geq
\textrm{IR}(\Delta)$, contradicting the minimality of $G=G(\Delta)$ according to
Definition~\ref{def:min-deviation}. To this end, we only need to show that the new
coalitional deviation $\Delta'$ takes place among the servers assigned with jobs
of the coalition $\Gamma'$, that is,
\begin{equation}\label{eqn:grand-inclusion}
    Q^+(\widetilde{M})\subseteq W \cup K' =
    \widetilde{W}_1\cup K\cup K',
\end{equation}
so that benefit of any job deviation will not decrease due to the fact that \emph{all} jobs on
servers of $X_1'$ migrate out in $\Delta$ and hence in $\Delta'$ as well,
leaving empty space for deviational jobs under $\Delta'$, which originally migrate to servers of
$\widetilde{W}_1$ under $\Delta$.

First we have $Q^+(K) \subseteq \widetilde{W}_1$ according to Lemma~\ref{lem:W1}.
On the other hand, it can be easily verified that $Q^{+}(X_1')
\subseteq W \bigcup K'$ according to the definition of $K'$. Now we show
$Q^{+}(K')\subseteq W $, which then
implies (\ref{eqn:grand-inclusion}). In fact, for any $k\in K'$, noticing that
$X_1' \subseteq X_1$, according to the definitions of $K'$ and $X_1$, we have
$k\in Q^+(X_{11})\backslash W$, which implies that $Q^+(k)\in W$ according to the
definition of $X_{11}$.
\end{proof}

\subsection{Part 2}

To prove our final upper bound, we need to introduce the following two structures in graph $G(\Delta)$
with tilde-valid arc set $\widetilde{A}$:
\begin{eqnarray*}
\Omega(\widetilde{A})&:=& \{i\in M':\ i_1 = \widetilde{Q}^{-1}(i)\in\widetilde{Q}^+(i), \\
& &    \ \delta^{-}(i_1)=1,\ \delta^{+}(i_1)=b_{i_1}>1\}; \\
\Pi(\widetilde{A}) &:=& \{(i,i_{1},j): \ i\in\widetilde{M}_{2}^{2}\backslash X_3,
    \ i_1\in \widetilde{Q}^+(\ell)\cap \widetilde{Q}^-(j)\backslash W, \\
& & \ i\ne j, \ \delta^{-}(i_1)=1,\ \delta^{+}(i_1)=b_{i_1}>1, \ j\in M'\}.
\end{eqnarray*}
Note that each element in $\Omega(\widetilde{A})$ represents a directed $2$-cycles of both arcs in
$\widetilde{A}$ and each element in $\Pi(\widetilde{A})$
is a directed 2-path of both arcs in $\widetilde{A}$. In both cases of
$\Omega(\widetilde{A})$ and $\Pi(\widetilde{A})$, the {interior} node $i_1$ has an in-degree
$\delta^-(i_1) = 1$ and all its out-arcs are in $\widetilde{A}$. Our next result is based on
the following further refinement of the tilde-valid arc set $\widetilde{A}$.
\begin{lemma}\label{lem:Omega}
If $\Omega(\widetilde{A})\ne\emptyset$ for some arc set $\widetilde{A}$ satisfying Assumption~\ref{ass:M-2-2},
then there exists an arc set $\widetilde{A}'$ such that, while it also satisfies Assumption~\ref{ass:M-2-2},
additionally, $\Omega(\widetilde{A}')$ is a {proper} subset of $\Omega(\widetilde{A})$.
\end{lemma}
\begin{proof}
Assume $i\in\Omega(\widetilde{A})$ and let $i_1= \widetilde{Q}^{-1}(i)$
be as in the definition of $\Omega(\widetilde{A})$. Then there must be a node
$h \in Q^{-}(i)$ with $h\neq i_1$, since otherwise $\delta^{-}(i)=\delta^{-}(i_1)=1$,
which implies that there would be no directed path from any other nodes in $G(\Delta)$ to nodes
$i$ or $i_1$, contradicting Lemma~\ref{lem:strong-connectivity}. Therefore, the following set is
not empty:
\begin{equation}\label{eqn:H}
    H_i:=\{h\in M:\ (h,i)\in A, \textrm{ either } \delta^-(h) > 1 \textrm{ or } (h,i)
\notin\widetilde{A}\}.
\end{equation}
Let $h \in H_i\ne\emptyset$. We define a new tilde-valid arc set
\begin{equation}\label{eqn:one-arc-swap}
    \widetilde{A}':=\{\widetilde{A}\backslash\{(i_{1},i)\}\}\cup\{(h,i)\}.
\end{equation}
It is easily seen that $i \in\Omega(\widetilde{A})\backslash\Omega(\widetilde{A}')$
and $\Omega(\widetilde{A}')\cup\{i\}=\Omega(\widetilde{A})$. On the other hand,
$\widetilde{A}'$ still satisfies Assumption~\ref{ass:mimimum-W0} due to $b_{i_{1}}>1$, and hence
also satisfies Assumption~\ref{ass:M-2-2} since $\delta^{+}(h)<a_{h}$ (which implies that
$h\notin \Omega(\widetilde{A})\cup\Omega(\widetilde{A}')$) according to
Observation~\ref{obs:out-degree-relation} (as no other node not in $M^2_2(\widetilde{A})$
can possibly become a member of $M^2_2(\widetilde{A}')$).
\end{proof}

As a result of Lemma~\ref{lem:Omega}, we can further refine our initial choice of
$\widetilde{A}$ so that it satisfies the following assumption, where the benefit of minimizing
$|\Pi(\widetilde{A})|$ will be seen in the proof of Lemma~\ref{lem:company} (see inequality (\ref{eqn:Pi})).
\begin{assumption}\label{ass:Omega-Pi}
Arc set $\widetilde{A}$ is such that it first satisfies Assumption~\ref{ass:M-2-2}
and then lexicographically minimizes
$(|\Omega(\widetilde{A})|,|\Pi(\widetilde{A})|)$.
\end{assumption}
\begin{corollary}\label{cor:Omega}
    Any arc set $\widetilde{A}$ satisfying Assumption~\ref{ass:Omega-Pi} must satisfy
    $\Omega(\widetilde{A})=\emptyset$. \ $\Box$
\end{corollary}

An arc set $\widetilde{A}$ in graph $G(\Delta)$ that satisfies Assumption~\ref{ass:Omega-Pi} is said to be
\emph{derived} from $\Delta$. Without loss of generality, our coalitional deviation
$\Delta$ is considered to have been chosen so that it satisfies the following assumption.

\begin{assumption}\label{ass:Delta}
Coalitional deviation $\Delta$ defining minimal deviation graph $G(\Delta)$ is such that
the arc set $\widetilde{A}$ derived from $\Delta$ gives lexicographical minimum
$V(\Delta):=(|W_{0}(\widetilde{A})|, |M_{2}^{2}(\widetilde{A})|,|\Omega(\widetilde{A})|,
|\Pi(\widetilde{A})|)$.
\end{assumption}

\begin{lemma}\label{lem:company}
Let minimal deviation graph $G(\Delta)$ with $\Delta$ satisfying Assumption~\ref{ass:Delta} be given.
For any $\ell \in X_6$, there is $j\ne \ell$, such that $\widetilde{Q}^{-1}(j)=\widetilde{Q}^{+}(\ell)\backslash W$
(note (\ref{eqn:bipartition})) and $j\in M''\backslash W$.
\end{lemma}
\begin{proof}
Given $\ell \in X_6$ and $\ell_1 = Q^+(\ell)\backslash W$. Since $a_{\ell_1}=b_{\ell_1}+1$
according to the definition of $X_6$, we have $\delta^+(\ell_1)=b_{\ell_1}$ and hence
$\widetilde{Q}^+(\ell_1)=Q^+(\ell_1)$ since $a_{\ell_1}=\delta^+(\ell_1)+1$ according to
Observation~\ref{obs:out-degree-relation}. Since $b_{\ell_1} > 1$ and
$\widetilde{Q}^+(\ell_1)=Q^+(\ell_1)\not\subseteq W$ (again according to the
definition of $X_6$), we let $j \in \widetilde{Q}^+(\ell_1)\backslash W$.
Then $j\ne \ell$ since otherwise we would
have $\ell \in \Omega(\widetilde{A})$, contracting Corollary~\ref{cor:Omega} with
our Assumption~\ref{ass:Omega-Pi} (see Fig.~\ref{fig:company} for an illustration with more explanations to follow).

\begin{figure}
\centering
\includegraphics[scale=0.95,viewport=40mm 210mm 165mm 255mm,clip=true]{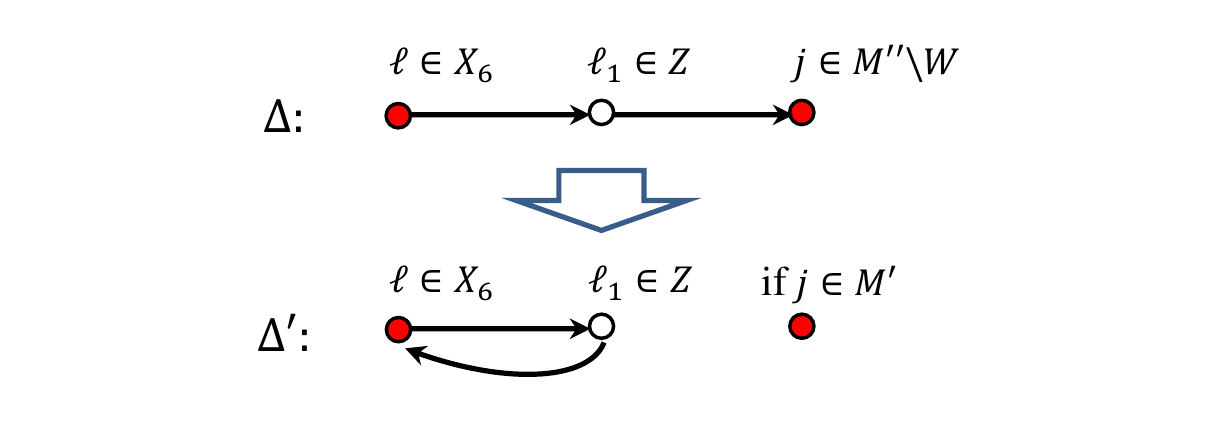}\\ %
  \caption{A Company Node $j$ of $\ell\in X_6$ or New Coalitional Deviation $\Delta'$}\label{fig:company}
\end{figure}

We claim $j\in M''$ and hence are done.
Let us assume for a contradiction that $j\in M'$. Note that
with $\{\ell,\ell_{1},j\}$ replacing $\{i,i_{1},j\}$ in the definition of $\Pi(\widetilde{A})$,
we conclude that $\ell\in \Pi(\widetilde{A})$. Now let us define a new
coalitional deviation $\Delta'$ so that its derived arc set $\widetilde{A}'$ gives a
$\mu(\Delta'):=(|W_{0}(\widetilde{A}')|, |M_{2}^{2}(\widetilde{A}')|, |\Omega(\widetilde{A}')|,
|\Pi(\widetilde{A}')|)$ that is lexicographically smaller than
$\mu(\Delta):=(|W_{0}(\widetilde{A})|,|M_{2}^{2}(\widetilde{A})|,|\Omega(\widetilde{A})|,
|\Pi(\widetilde{A})|)$, a desired contraction to Assumption~\ref{ass:Delta}.

In fact, let $\Delta'$ be defined as in Observation~\ref{obs:equiv-min-deviation}
after node $i$ has been replaced by $\ell$ in the statement of
Observation~\ref{obs:equiv-min-deviation}. Denote $A'$ as the arc set of the resulting
minimal deviation graph $G(\Delta')$. Let $\widetilde{A}'$ be the natural
result of $\widetilde{A}$ after the re-orientation from $\Delta$ and $\Delta'$, i.e.,
an arc in $\widetilde{A}$ pointing to $\ell$ (resp.~$j$) will become an arc in
$\widetilde{A}'$ pointing to $j$ (resp.~$\ell$). Other arcs are the same for
$\widetilde{A}$ and $\widetilde{A}'$. Apparently,
\[
|W_{0}(\widetilde{A}')|=|W_{0}(\widetilde{A})|,\
|M_{2}^{2}(\widetilde{A}')|=|M_{2}^{2}(\widetilde{A})|.
\]
On the other hand, if the value of $|\Omega(\widetilde{A}')|$ has increased from
$|\Omega(\widetilde{A})|$, then clearly it must be the result of $\ell$ and/or $j$
becoming element(s) of $\Omega(\widetilde{A}')$. In any such case (say, the former
case for the sake of argument), based on the definition of $\Omega(\widetilde{A}')$, we
can use the approach in Lemma~\ref{lem:Omega} to find $h\in H_{\ell}$ as defined in
(\ref{eqn:H}) and perform an arc-swap as in
(\ref{eqn:one-arc-swap}) with $i$ and $i_1$ replaced by $\ell$ and $\ell_1$,
respectively, to reduce $|\Omega(\widetilde{A}')|$ while maintaining
the values of $|W_{0}(\widetilde{A}')|$ and $|M_{2}^{2}(\widetilde{A}')|$. For
convenience, we still use $\widetilde{A}'$ to denote the tilde-valid arc set after such
arc-swap(s) if needed. Consequently, we have
\[
\Omega(\widetilde{A}') = \Omega(\widetilde{A}) = \emptyset.
\]
However, we claim
\begin{equation}\label{eqn:Pi}
    |\Pi(\widetilde{A}')|<|\Pi(\widetilde{A})|,
\end{equation}
a desired contradiction. To see inequality (\ref{eqn:Pi}), we first note that
(i) any 2-path in $\Pi(\widetilde{A})$ starting
at $i\ne\ell, j$ is also a 2-path in $\Pi(\widetilde{A}')$, and vice versa, and (ii) any 2-path
in $\Pi(\widetilde{A})$ (resp.~$\Pi(\widetilde{A}')$) starting at $\ell$ (resp.~$j$) must
have the first arc $(\ell, \ell_1)$ (resp.~$(j, \widetilde{Q}^+(j)\backslash W)$,
since $|\widetilde{Q}^+(j)\backslash W| =1$ due to $j\in \widetilde{M}_2^2\backslash X_3$ according to
(\ref{eqn:bipartition})).
On the other hand, the following can be easily observed:
\begin{enumerate}
  \item If $(\ell, \ell_1, j')\in \Pi(\widetilde{A})$ ($j'\ne j$), then
        $(\ell, \ell_1, j')\in \Pi(\widetilde{A}')$, and vice versa.
  \item If $(j, j_1, j')\in \Pi(\widetilde{A})$ ($j'\ne \ell$), then
        $(j, j_1, j')\in \Pi(\widetilde{A}')$, and vice versa.
  \item $(\ell, \ell_1, j)\in \Pi(\widetilde{A})\backslash \Pi(\widetilde{A}')$,
        since $(\ell, \ell_1, j)\in \Pi(\widetilde{A}')$ would imply
        $(\ell_1, j)\in\widetilde{A}'\subseteq A'$ by definition of $\Pi(\widetilde{A}')$
        and hence $(\ell_1, \ell)\in A$ by definition of $A'$, which in turn
        implies that $(\ell_1, \ell)\in \widetilde{A}$ since $b_{\ell_1}=\delta^+(\ell_1)$
        under $\widetilde{A}$. Consequently, we obtain
        $\ell\in \Omega(\widetilde{A})$, contradicting
        Corollary~\ref{cor:Omega}.
  \item With similar reasons for $(\ell, \ell_1, j)\not\in\Pi(\widetilde{A}')$,
        we have $(j, j_1, \ell) \not\in \Pi(\widetilde{A}')$.
\end{enumerate}
Therefore, overall $\Pi(\widetilde{A}')$ contains at least one element less than $\Pi(\widetilde{A})$
as indicated in points 3 and 4 above.
\end{proof}

We call $j\in M''\backslash W$ identified in the above lemma a \emph{company} of $\ell\in X_6$. Clearly, any $j\in M''\backslash W$ cannot be a company of two different elements of $X_6$ according to the statement of the lemma, which leads us to the following corollary.
\begin{corollary}\label{cor:X6}
Denote $X[\ell]:= \{j\in M\backslash X:\ \textrm{node $j$ is a company of $\ell$}\}$ for any $\ell\in X_6$ and let
$Y_{6}:=\bigcup_{\ell\in X_6}((\widetilde{Q}^{+}(\ell)\backslash{W})\cup X[\ell])$.
Then $Y_6\subseteq M''\backslash W$ and $2|X_{6}|\leq|Y_{6}|$. \ $\Box$
\end{corollary}

We have used the cardinalities of the six sets $Y_1, \dots, Y_6$ to bound $|X_1|$, $|X_2|$, $2|X_3|$, $|X_4|$, $|X_5|$ and $2|X_6|$, respectively. Let us make sure these sets do not overlap with $X$ and are mutually disjoint. According to Lemmas~\ref{lem:X2}--\ref{lem:X5} and Corollary~\ref{cor:X6}, we have
\begin{align*}
&   Y_5\subseteq M_2^1\cap M'\backslash W\subseteq M\backslash X;\\
&   Y_2, Y_3, Y_4, Y_6 \subseteq M''\backslash W\subseteq M\backslash X;
\end{align*}
and
\begin{align*}
&   Y_t\subseteq \widetilde{Q}^+(X_t), t=2, 3, 4; \\
&   Y_6\subseteq \widetilde{Q}^+(X_6)\cup\widetilde{Q}^+(M\backslash X).
\end{align*}
Hence $W\cap X = \emptyset$ and $W\cap Y_t = \emptyset$ ($t=2, \ldots, 6$). Since $\widetilde{Q}^+(i)\cap \widetilde{Q}^+(j)=\emptyset$ for any $i\neq j$ (definition of $\widetilde{A}$) and $\widetilde{Q}^+(i)\cap M'=\emptyset$ for any $i$ (Observation~\ref{obs:out-degree-relation}),
noticing that $Y_1=W_1\subseteq W\backslash X$ (Lemma~\ref{lem:W1}), we conclude that
\begin{equation}\label{eqn:disjoint-Y}
    Y_t\cap X = \emptyset \textrm{ and }  Y_t\cap Y_s = \emptyset, \ s \ne t,
    \ s,t     \in \{1, \ldots, 6\}.
\end{equation}

\section{Establishment of Strong Stability}

Now we are ready to go back to proving (\ref{eqn:ultimate}) and hence Theorem~\ref{thm:main}.
Since $X_2\subseteq M_3^3$, the left-hand side of inequality (\ref{eqn:ultimate}) is at most
\begin{equation}\label{eqn:LHS}
    |X|-\frac12|M_3^3| \le \sum_{t=1}^6 |X_t|-\frac12|X_2|.
\end{equation}
On the other hand, if we let
\[
Y_t':=\{i\in Y_t:\ b_i=1\} \ \textrm{ and } \
Y_t'':=Y_t\backslash Y_t', \textrm{ for } t=2,3,4,6,
\]
which imply
\begin{align*}
    & Y_4''=\{i\in Y_4:\ b_i\ge 2, a_i\ge b_i+2\}
              \subseteq\bigcup_{2\le b\le a-2} M_a^b,  \\
    &\bigcup_{t\in\{2,3,4,6\}}Y_t'\cup Y_1\cup Y_5
              \subseteq \bigcup_{a\ge 2} M_a^1, \\
    &\bigcup_{t\in\{2,3,6\}}Y_t''
              \subseteq\bigcup_{2\le b < a} M_a^b,
\end{align*}
then noticing the properties (\ref{eqn:disjoint-Y}) and that
\begin{equation*}
    b-\frac{4b-5}{a-1} \ge \left\{
    \begin{array}{ll}
      \frac12, & \textrm{if } 2\le b <a; \\
      1, & \textrm{if } 2\le b \le a-2,
    \end{array}
\right.
\end{equation*}
we see that the right-hand side of inequality (\ref{eqn:ultimate}) is at least
\begin{align}
    |Y_1|+|Y_5| & +\sum_{t\in\{2,3,4,6\}} |Y_t'|  +
                  \frac12\sum_{t\in\{2,3,6\}} |Y_t''| + |Y_4''| \nonumber \\
                & \ge |Y_{1}|+\frac{1}{2}|Y_{2}|+\frac{1}{2}|Y_{3}|+
                  |Y_{4}|+|Y_{5}|+\frac{1}{2}|Y_{6}|. \label{eqn:middle}
\end{align}
According to Lemmas~\ref{lem:X2}--\ref{lem:X1} and Corollary~\ref{cor:X6}, the right-hand side of inequality (\ref{eqn:LHS}) is at most that of (\ref{eqn:middle}), which
in turn ultimately leads to inequality (\ref{eqn:ultimate}). Consequently, Theorem~\ref{thm:main} is established.

\medskip

From Theorem~\ref{thm:main} and the lower bound demonstrated in Section~\ref{sec:lower_bound},
the following theorem follows.
\begin{theorem}
In the m-server load balancing game ($m\geq3$), any NE is a
$({5}/{4})$-approximate SNE and the bound is tight.
\end{theorem}

\section{Concluding Remarks}

By establishing a tight bound of $5/4$ for the approximation of general NEs to SNEs in the
$m$-server load balancing game for $m\geq3$, we have closed the final gap for the study of
approximation of general NEs to SNEs. However, as demonstrated by Feldman \&
Tamir \cite{FeTa09} and by Chen \cite{Chen09}, a special subset of NEs known as LPT
assignments, which can be easily identified as NEs \cite{FoKoMaSp02},
do approximate SNEs better than general NEs. It is still a challenge
to provide a tight approximation bound for this subset of NEs.

\section*{Acknowledgement}
Research by the second and third author was partially supported by the National Natural Science
Foundation of China (Grant No.~11071142) and the Natural Science
Foundation of Shandong Province, China (Grant No.~ZR2010AM034).

\end{document}